\documentclass[a4paper,11pt]{article}
\usepackage{graphicx}
\usepackage{amsmath}
\usepackage{amssymb}
\usepackage{subfigure,fullpage,dsfont,url}
\usepackage[ansinew]{inputenc}
\usepackage[hang,center]{caption}

\graphicspath{{figures/}}
\newcommand{\qed}{\hspace{\stretch{1}}$\Box$}

\newcommand{\R}{\ensuremath{\mathds R}}

\newcommand{\Z}{\ensuremath{\mathds Z}}
\newtheorem{lemma}[equation]{Lemma}
\newtheorem{theorem}[equation]{Theorem}
\newtheorem{obs}[equation]{Observation}

\newenvironment{proof}{\vspace{-.25\baselineskip}\noindent\textbf{Proof.}
}{\qed\par\medskip}

\title{Consistent digital line segments}
\author{{Tobias Christ\thanks{Institute of Theoretical Computer Science, ETH Zurich, Switzerland, email: tobias.christ@inf.ethz.ch}} \quad
{D\"om\"ot\"or P\'alv\"olgyi\thanks{Department of Combinatorial Geometry, EPF Lausanne, Switzerland, email: dom@cs.elte.hu.}} \quad {Milo\v{s} Stojakovi\'{c} \thanks{Department
of Mathematics and Informatics, University of Novi Sad, Serbia, email: milos.stojakovic@dmi.uns.ac.rs.}}}
\begin{document}

\maketitle

\begin{abstract}
We introduce a novel and general approach for digitalization of line segments in
the plane that satisfies a set of axioms naturally arising from Euclidean
axioms. In particular, we show how to derive such a system of digital segments
from any total order on the integers. As a consequence, using a well-chosen
total order, we manage to define a system of digital segments such
that all digital segments are, in Hausdorff metric, optimally close to their
corresponding Euclidean segments, thus giving an explicit construction that
resolves the main question of~\cite{CKNT09}.
\end{abstract}

\section{Introduction}

%\dom{we would need a better/shorter name than CDLSS}

%\mSt{I think that CDLSS is not bad. Especially that I can think of nothing better. And it's fairly common to abbreviate this way.}

%\dom{okay, I don't insist on changing it but a few letters could be really dropped. how about DLS? (DSS seems to be in use already for digital straight segment)}

%\tobi{I decided to use CDS, as Milo\v{s} suggested in one of the mails I guess.}

%\mSt{Add your email addresses.}

One of the most fundamental challenges in digital geometry is to define a
``good'' digital representation of a geometric object. Of course, the meaning of
the word ``good'' here heavily depends on particular conditions we may impose.
Looking at the problem of digitalization in the plane, the goal is to find a set
of points on the integer grid $\Z^2$ that approximates well a given object. The
topology of the grid $\Z^2$ is commonly defined by the graph whose vertices are
all the points of the grid, and each point is connected by an edge to each of
the four points that are either horizontally or vertically adjacent to it.

Knowing that a straight line segment is one of the most basic geometric objects
and a building block for many other objects, defining its digitalization in a
satisfying manner is vital. Hence, it is no wonder that this has been a hot
scientific topic in the last few decades,
see~\cite{KR04} for a recent survey and
\cite{GDYF86}, \cite{Goodrich97snaprounding}, and \cite{Sugihara01}
for related work, dealing with the problem of representing objects in digital geometry without causing topological and combinatorial inconsistencies.

For any pair of points $p$ and $q$ in the grid
$\Z^2$ we want to define the digital line segment $S(p,q)$ connecting them, that
is, $\{p,q\} \subseteq S(p,q) \subseteq \Z^2$.
Chun et al.\ in~\cite{CKNT09} put forward the following four axioms that arise naturally from
properties of line segments in Euclidean geometry.

\begin{enumerate}

\item[(S1)] \emph{Grid path property:} For all $p, q \in \Z^2$, $S(p,q)$ is the
vertex set of a path from $p$ to $q$ in the grid graph.

\item[(S2)] \emph{Symmetry property:} For all $p, q \in \Z^2$, we have $S(p,q) =
S(q,p)$.

\item[(S3)] \emph{Subsegment property:} For all $p, q \in \Z^2$ and every $r \in
S(p,q)$, we have $S(p,r) \subseteq S(p,q)$.

\item[(S4)] \emph{Prolongation property:} For all $p,q \in \Z^2$, there exists
$r \in \Z^2$, such that $r \notin S(p,q)$ and $S(p,q) \subseteq S(p,r)$.

\end{enumerate}

First, note that (S3) is not satisfied by the usual way a computer visualizes a
segment. A natural definition of the digital straight segment between $p=(p_x,p_y)$ and
$q=(q_x,q_y)$, where $p_x\le q_x$ and $0\le q_y-p_y<q_x-p_x$ is
$\left\{\left(x,\left\lfloor (x-p_x)\frac{q_y-p_y}{q_x-p_x}+p_y+0.5\right\rfloor\right):\,
p_x \leq x \leq q_x \right\}$. This does not satisfy (S1), but it could be
easily fixed by a slight modification of the definition. Still, it also does not satisfy
(S3), for example, for $p=(0,0), r=(1,0), q=(4,1)$, the subsegment from $r$ to $q$ is not
contained in the segment from $p$ to $q$.

%\mSt{Is the last definition really commonly known as "the usual way a computer visualizes
%a segment?" If not, please rephrase with a lest bombastic statement!}

%\dom{changed}

Even though the set of axioms (S1)-(S4) seems rather natural, there are still some fairly
exotic examples of digital segment systems that satisfy all four of them. For
example, let us fix a double spiral $\cal D$ centered at an arbitrary point of
$\Z^2$, traversing all the points of $\Z^2$. As it is a spanning path of the
grid graph, we can set $S(p,q)$ to be the path between $p$ and $q$ on $\cal D$,
for every $p,q\in\Z^2$. It is easy to verify that this system satisfies axioms
(S1)-(S4).

Another condition was introduced in~\cite{CKNT09} to enforce the monotonicity of
the segments, ruling out pathological examples like the one above. Here, we
phrase this monotonicity axiom differently, but still, the system of axioms
(S1)-(S5) remains equivalent to the one given in~\cite{CKNT09}.

\begin{enumerate}
\item[(S5)]
\emph{Monotonicity property:} If both $p, q \in \Z^2$ lie on a line that is
either horizontal or vertical, then the whole segment $S(p,q)$ belongs to this line.
\end{enumerate}

We call a system of digital line segments that satisfies the system of axioms
(S1)-(S5) a \emph{consistent digital line segments system (CDS)}. It is
straightforward to verify that every CDS also satisfies the following three
conditions.

%%\mSt{Do we want the second condition, knowing that we have the first?}
%%\dom{I would keep all three}
%%\tobi{I agree.}

\begin{itemize}
\item[(C1)] If the slope of the line going through $p$ and $q$ is non-negative,
then the slope of the line going through any two points of $S(p,q)$ is non-negative.
The same holds for non-positive slopes.
\item[(C2)] For all $p, q \in \Z^2$, the grid-parallel box spanned by points $p$
and $q$ contains $S(p,q)$.
\item[(C3)] If the intersection of two digital segments contains two points $p, q
\in \Z^2$, then their intersection also contains the whole digital segment
$S(p,q)$.
\end{itemize}

We give a simple example of a CDS, where the segments follow the boundary of the
grid-parallel box spanned by the endpoints. Let $p, q\in \Z^2$ be two points
with coordinates $p=(p_x,p_y)$ and $q=(q_x,q_y)$. If $p_y\leq q_y$, we
 define $S(p,q)=S(q,p)=\{(x,p_y):\, \min\{p_x,q_x\} \leq x
\leq \max\{p_x,q_x\} \} \cup \{(q_x,y):\, p_y \leq y \leq q_y\} \}$.
If $p_y > q_y$, we swap the points $p$ and $q$, and define the segment as in
the previous case.

It can be easily verified that this way we defined a CDS, but the digital
segments in this system visually still do not resemble well the Euclidean
segments.

One of the standard ways to measure how close a digital segment is to a
Euclidean segment is to use the Hausdorff distance. We denote by $\overline{pq}$
the Euclidean segment between $p$ and $q$, and by $|\overline{pq}|$ the
Euclidean length of $\overline{pq}$. For two plane objects $A$ and $B$, by $H(A,B)$ we
denote their Hausdorff distance.

The main question raised in~\cite{CKNT09} was if it is possible to define a CDS
such that a Euclidean segment and its digitalization have a reasonably small
Hausdorff distance. More precisely, the goal is to find a CDS satisfying the
following condition.

\begin{enumerate}
\item[(H)] \emph{Small Hausdorff distance property:} For every $p, q \in \Z^2$,
we have that $H(\overline{pq},S(p,q))= O(\log |\overline{pq}|)$.
\end{enumerate}

Note that in the CDS example we gave, the Hausdorff distance between a Euclidean
segment of length $n$ and its digitalization can be as large as $n/\sqrt{2}$.

While this question was not resolved in~\cite{CKNT09}, a clever construction of
a system of digital rays \emph{emanating from the origin of $\Z^2$} that satisfy
(S1)-(S5) and (H) was presented. Moreover, it was shown using Schmidt's theorem~\cite{Sch72}
that already
for rays emanating from the origin, the $\log$-bound imposed in condition (H) is
the best bound we can hope for, directly implying the following theorem.

\begin{theorem} \textnormal{\cite{CKNT09}} \label{t:cknt}
There exists a constant $c>0$, such that for any CDS and any $d > 0$,
there exist $p, q \in \Z^2$ with $|\overline{pq}| > d$, such that
$H(\overline{pq},S(p,q))> c\log |\overline{pq}|$.
\end{theorem}

In this paper, we introduce a novel and general approach for the construction of
a CDS. Namely, for any total order $\prec$ on $\Z$, we show how to derive a CDS
from $\prec$. (By total order we always mean a strict total order.)
This process is described in Section~\ref{s:def}. As a
consequence, in Section~\ref{s:H}, we manage to define a CDS that satisfies (H),
deriving it from a specially chosen order on $\Z$, and thus giving the explicit
construction that resolves the main question of~\cite{CKNT09}.

%\tobi{I guess strictly speaking we are using \emph{strict total orders} which is not really the same but somehow equivalent to a total order. Right?}

%\mSt{You are right. But as strict total orders and total orders are basically the same structure, I think it's clear what we mean. To distinguish between the two, it's quite enough to write $\prec$ or $\preceq$. So, I would leave the comment in this paragraph (to be on the safe side), and wouldn't mention ``strict'' anywhere else (I removed it from the abstract).}

\begin{theorem} \label{t:main}
There is a CDS that satisfies condition (H).
\end{theorem}

Note that Theorem~\ref{t:cknt} ensures that such a CDS is optimal up to a constant factor in terms of
the Hausdorff distance from the Euclidean segments. In Section~\ref{s:char}
we make a step towards a characterization of CDSes, demonstrating their natural connection to total orders on $\Z$, while
in Section~\ref{s:lines} we make use of the digital line segment definition to introduce digital lines. Finally, in
Section~\ref{s:high-d} we discuss prospects of using a similar approach to define a CDS in higher dimensions.

\section{Digital line segments derived from a total order on $\Z$} \label{s:def}

Let $\prec$ be a total order on $\Z$. We are going to define a CDS ${\cal
S}_\prec$, deriving it from $\prec$.

Let $p,q \in \Z^2$, $p = (p_x,p_y)$ and $q
= (q_x,q_y)$. If $p_x>q_x$, we swap $p$ and $q$. Hence, from now on we may
assume that $p_x \leq q_x$.

If $p_y \leq q_y$, then $S_{\prec}(p,q)$ is defined as follows. We start at the
point $p=(p_x,p_y)$ and we repeatedly go either up or to the right, collecting
the points from $\Z^2$, until we reach $q$. Note that the sum of the coordinates $x+y$
increases by $1$ in each step. In total we have to make $q_x+q_y-p_x-p_y$ steps
and in exactly $q_y - p_y$ of them we have to go up.
The decision whether to go up or to the right is made as follows: if we are at the
point $(x,y)$ for which $x+y$ is among the $q_y-p_y$ greatest elements of the
interval $[p_x+p_y,q_x+q_y-1]$ according to $\prec$, we go up, otherwise we go
to the right. We will refer to this interval as the \emph{segment interval}.

If $p_y > q_y$, that is, if $p$ is the top-left and $q$ the bottom-right corner
of the grid-parallel box spanned by $p$ and $q$, then we define $S_{\prec}(p,q)$
as the mirror reflection of $S_{\prec}((-q_x,q_y),(-p_x,p_y))$ over the
$y$-axis.

\noindent
{\bf Example.} Suppose $p=(0,0)$ and $q=(2,2)$. Their segment interval consists
of four numbers, $0,1,2,3$. If $\prec$ is the natural order on $\Z$, then the two
greatest elements of the segment interval are $2$ and $3$. Since $0+0$ is not
one of these, at $(0,0)$ we go right, to $(1,0)$. At $(1,0)$ we again go to
right, to $(2,0)$, from there to $(2,1)$ (since $2+0$ is one of the greater
elements) and finally to $(2,2)$. In fact, it can be easily seen that using the
natural order on $\Z$
we get the CDS mentioned in Section 1, the one that always follows the boundary of
the box spanned by the endpoints.

\begin{theorem} \label{theorem_order_to_CDS}
${\cal S}_\prec$, defined as above, is a CDS.
\end{theorem}

\begin{proof}
We will verify that ${\cal S}_\prec$ satisfies the axioms (S1)-(S5).
\begin{itemize}

\item[(S1)] The condition (S1) follows directly from the definition of ${\cal
S}_\prec$.

\item[(S2)] Let $p,q$ be two points from $\Z^2$. If the first coordinates of $p$ and
$q$ are different, then condition (S2) follows directly. Otherwise, $p$ and $q$
belong to the same vertical line, and from the construction we see that both
$S_{\prec}(p,q)$ and $S_{\prec}(q,p)$ consist of all the points on that line
between $p$ and $q$.

\item[(S3)] For a contradiction, assume that there are points $p=(p_x,p_y)$,
$q=(q_x,q_y)$ and $r=(r_x,r_y)$, with $r\in S_{\prec}(p,q)$, such that
$S_{\prec}(p,r) \not\subseteq S_{\prec}(p,q)$.
W.l.o.g.\ we may assume that $\overline{pq}$ has a non-negative slope.

\emph{Case 1.} $p_x\leq q_x$ and $p_y\leq q_y$.
We also have $p_x\leq r_x$ and $p_y\leq r_y$, and going on each of the segments
$S_{\prec}(p,r)$ and $S_{\prec}(p,q)$ point-by-point starting from $p$, we move
either up or right. By assumption, these two segments separate at some point
$(a,b)$ and then meet again, for the first time after this separation, at some
other point $(c,d)$, see Figure~\ref{fig_subsegment_property}. One of the
segments goes up at $(a,b)$ and enters $(c,d)$ horizontally coming from the
left, which implies that $a+b$ is among the greater numbers of the segment
interval of this segment, while $c+d-1$ is not, thus $c+d-1 \prec a+b$. But the
other segment goes horizontally at $(a,b)$ and enters $(c,d)$ vertically coming
from below, which similarly implies $ a + b \prec c + d - 1$, a contradiction.

\begin{figure}
\centering
\includegraphics{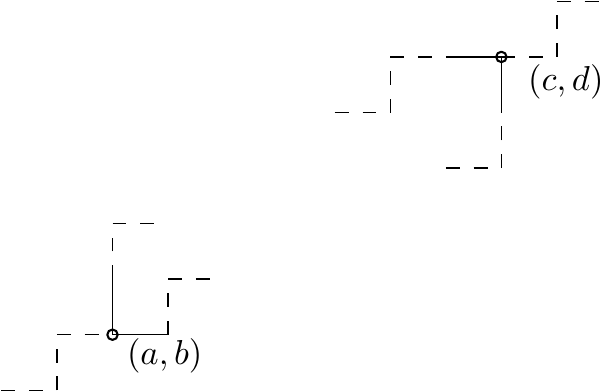}
\caption{Two paths splitting up at $(a,b)$ and meeting again at $(c,d)$.
\label{fig_subsegment_property}}
\end{figure}

\emph{Case 2.} $q_x\leq p_x$ and $q_y\leq p_y$.
We also have $q_x\leq r_x$ and $q_y\leq r_y$. By assumption, the two segments
starting at $q$ and $r$, $S_{\prec}(q,p)$ and $S_{\prec}(r,p)$, separate at some point $(a,b)$ and then meet again, for
the first time after this separation, at some other point $(c,d)$. Using the
same argument as before, we get a contradiction. Hence, (S3) holds.

\item[(S4)] To show that condition (S4) holds, consider the segment from $p=(p_x,p_y)$
to $q=(q_x,q_y)$. W.l.o.g.\ we can assume that $p_x \leq q_x$ and $p_y \leq
q_y$. We distinguish two cases.

\emph{Case 1.} If $q_x+q_y$ is among the $q_y-p_y+1$ greatest numbers of $[p_x+p_y,q_x+q_y]$
according to $\prec$, then we can prolong the segment going one step vertically
up, that is, the segment $S_{\prec}((p_x,p_y),(q_x,q_y+1))$ contains the segment
$S_{\prec}((p_x,p_y),(q_x,q_y))$ as a subsegment.

\emph{Case 2.} If, on the other hand,
$q_x+q_y$ is not among the $q_y-p_y$ greatest numbers of $[p_x+p_y,q_x+q_y]$, we
can prolong the segment horizontally to the right, that is,
$S_{\prec}((p_x,p_y),(q_x,q_y)) \subset S_{\prec}((p_x,p_y),(q_x+1,q_y))$.

Note that if $q_x+q_y$ is exactly the $(q_y-p_y+1)^{th}$ number in
$[p_x+p_y,q_x+q_y]$, then the conclusions of both cases are true, and indeed the
rays emanating from $(p_x,p_y)$ split at $(q_x,q_y)$.

\item[(S5)] The condition (S5)
follows directly from the definition of ${\cal S}_\prec$.
\end{itemize}
\vspace{-10mm}
\end{proof}

Apparently in the definition of $S_{\prec}(p,q)$ only the sum of the coordinates of the points
plays a role, so if we translate $p$ and $q$ by a vector $(t,-t)$, for any integer $t$, the
digital line segment will look the same.

\begin{obs} \label{obs_diagonal_translation_invariance}
    Let $t \in \Z$ be an integer. $S_{\prec}(p+(t,-t),q+(t,-t)) = S_{\prec}(p,q) + (t,-t) = \{(x+t,y-t) \in \Z^2 : (x,y) \in S_{\prec}(p,q) \}$.
\end{obs}

\section{Digital segments with small Hausdorff distance to Euclidean
segments} \label{s:H}

%\mSt{I changed the definition here, hopefully to a correct one. Please check carefully, we don't want a wrong one...}

%\dom{seems ok to me}

For integers $k$ and $l\geq 2$, let $|k|_l$ denote the number of times $k$ is
divisible by $l$, that is, $$|k|_l = \sup\left\{m:\, l^m \,|\, k\right\}.$$

We define a total order on $\Z$ as follows. Let $a \prec b$ if and only if there
exists a non-negative integer $i$ such that $|a-i|_2 < |b-i|_2$, and for all
$j\in\{0,\dots, i-1\}$ we have $|a-j|_2 = |b-j|_2$. In plain words, for two
integers $a$ and $b$, we say that the one that contains a higher power of $2$ is
greater under $\prec$. In case of a tie, we repeatedly subtract 1 from both $a$
and $b$, until at some point one of them contains a higher power of $2$ than the
other. Thus, for example, $-1\prec -5\prec 3\prec -3\prec 5\prec 1\prec -2\prec 6\prec -6\prec 2\prec -4\prec 4\prec 0$.

%\mSt{The way the following paragraph is now, it's not clear to me what is the connection with
%the order we just defined. There should be either an additional explanation, or it should be moved/removed.}
%
%\dom{Why, because previously I did not put $\prec$-decreasing but just decreasing? Should I add an example?}
% Milos: All is well now.

Note that if we take the elements of an interval of the form $(-2^n,2^n)$ in $\prec$-decreasing
order and we apply the function $0.5-x2^{-n-1}$ to them, then we get the first few elements of the Van der
Corput sequence~\cite{vdC35}.

\begin{figure}
\centering
\vspace{-1mm}
\includegraphics[width=0.9\textwidth]{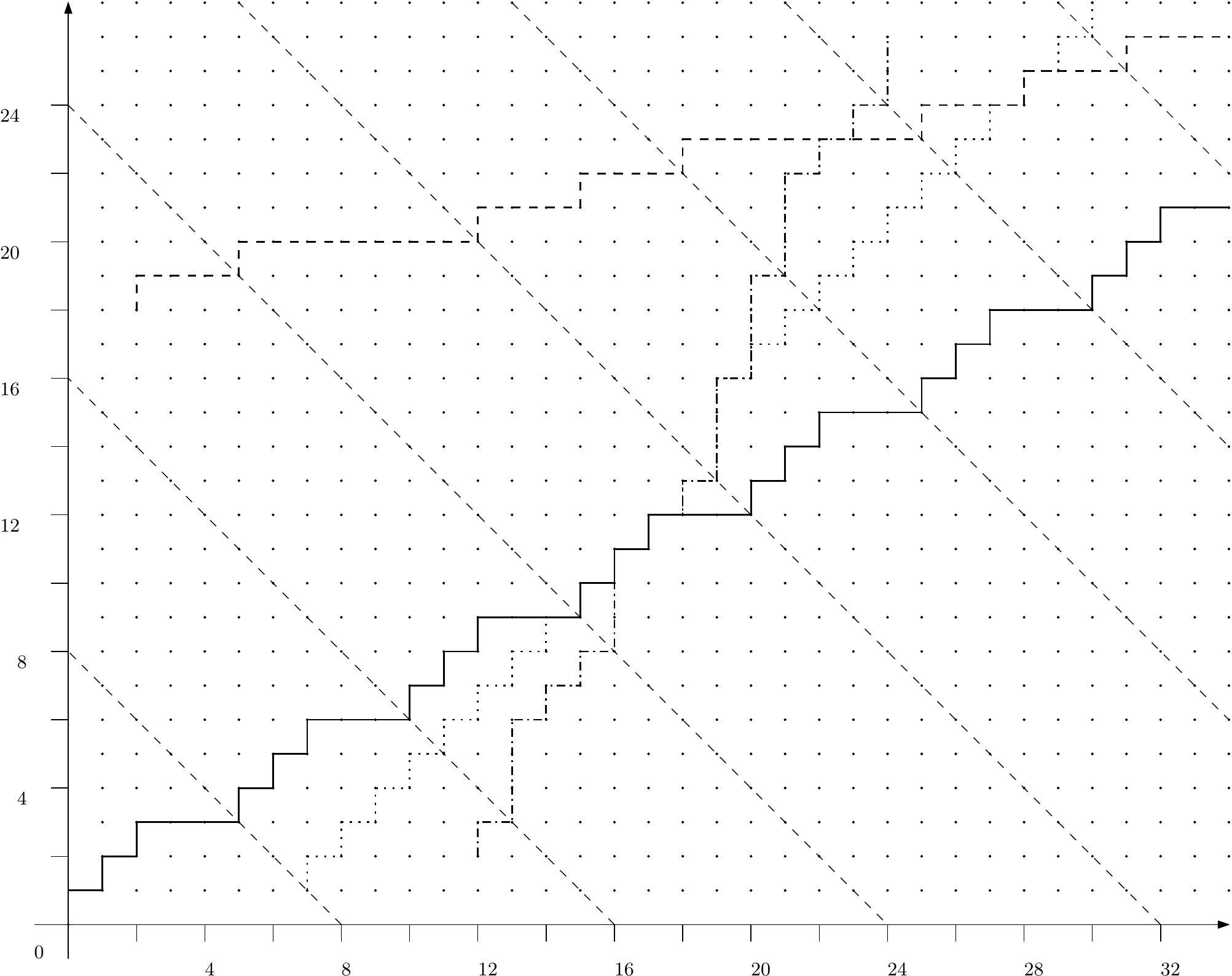}
\caption{Some line segments.\label{fig_long_segments}}
\vspace{-1mm}
\end{figure}

We will prove that using this total order to define the system of digital line segments ${\cal
S}_\prec$, as described in the previous section, we obtain a CDS which
%turns out
%to have Hausdorff distance between the digital segment $S_\prec(p,q)$ and its
%corresponding Euclidean segment $\overline{pq}$ which is at most logarithmic in
%$|\overline{pq}|$, i.e., it
satisfies condition (H). In
Figure~\ref{fig_long_segments} we give some examples of digital segments in this CDS, and
Figure~\ref{fig_rays_from_0_and_more} shows the segments emanating from $(0,0)$ to
some neighboring points, as well as the segments from $(2,3)$ to the neighboring points.

At first sight it may be
surprising to observe that all the segments emanating from the origin in our
construction coincide with the ones given in the construction of digital rays
in~\cite{CKNT09}. However, it is not a coincidence, as the construction
from~\cite{CKNT09} also relies on the same total order on integers.

\begin{figure}
\centering
\includegraphics{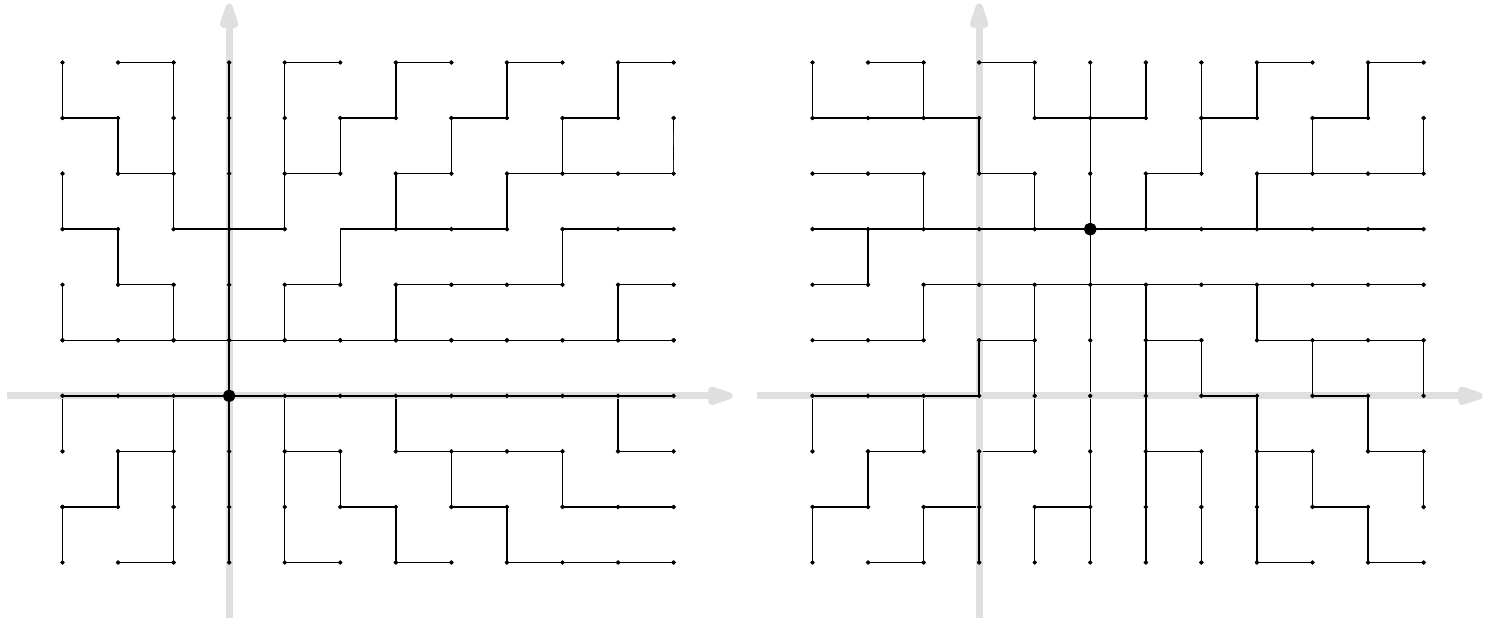}
\caption{Digital line segments emanating from $(0,0)$ and from $(2,3)$. \label{fig_rays_from_0_and_more}}
\end{figure}

%\tobi{Here's my suggestion for a proof. I'm not claiming that it's the easiest way to do it, but it's what's closest
%to my intuition. It's somewhat ugly as it has this constant $c = \sqrt{5}/2$ in it, but I'm quite convinced that the
%statement holds without, i.e., $H(\overline{pq},S_\prec(p,q)) \leq \log(p_x+p_y-q_x-q_y) + 1$ should be provable. But right
%now I don't see an easy way to get rid of it.}

For points $v,w \in \R^2$ and $A \subseteq \R^2$, let $d(v,w) = |v-w|$ and $d(v,A) = \inf_{a \in A} d(v,a)$ denote the
usual Euclidean distances between two points, and between a point and a set. For $p,q,r,s \in \Z^2$, by
$\overline{pqrs}$ we denote the union of Euclidean linear line segments from $p$ to $q$, from $q$ to $r$, and
from $r$ to $s$.

\begin{obs}
For any $p,q \in \Z^2$, $H(S_\prec(p,q),\overline{pq}) =
max\{d(r,\overline{pq}) : r \in S_\prec(p,q) \}$.
\end{obs}

We proceed by proving three statements that we will use to ultimately prove Theorem~\ref{t:main}.

\begin{lemma} \label{lemma:subs}
If $p,q \in \Z^2$ and $ r,s \in S_\prec(p,q)$. Then
\[
H(\overline{rs},S_\prec(r,s))
\leq 2H(\overline{pq},S_\prec(p,q)).
\]
\end{lemma}

\begin{proof} We know that $d(r,\overline{pq}) \leq H(\overline{pq},S_\prec(p,q)) =:
h$ and $d(s,\overline{pq}) \leq h$, therefore $H(\overline{prsq},\overline{pq})
\leq h$. Hence,  for all $v \in \overline{pq}$, $d(v,\overline{prsq}) \leq h$.
Let $t \in S_\prec(r,s) \subseteq S_\prec(p,q)$ and $v \in
\overline{pq}$ be such that $d(t,v) = d(t,\overline{pq}) \leq h$. Using the
triangle inequality we conclude $d(t,\overline{prsq}) \leq d(t,v) +
d(v,\overline{prsq}) \leq 2h$. Because of (C2), we have $d(t,\overline{rs}) =
d(t,\overline{prsq}) \leq 2h$, and therefore $H(S_\prec(r,s),\overline{rs}) \leq 2h$.
\end{proof}

\begin{lemma} \label{lemma_c}
Let $p,q,r,r' \in \Z^2$, such that $r_x - p_x =
q_x - r'_x + \varepsilon$, $r_y - p_y = q_y - r_y' - \varepsilon$, with
$\varepsilon \in\{ 0, 1,- 1\}$, $r'_x = r_x$ and $r'_y = r_y+1$. Then
$H(\overline{pq},\overline{prr'q}) \leq c = \sqrt{5}/2$.
\end{lemma}

\begin{proof} Without loss of generality $p = (0,0)$. We have
\[
  H(\overline{pq},\overline{prr'q}) = \max\{ d(r,\overline{pq}),d(r',\overline{pq}) \}.
\]
By assumption $ 2r_x - \varepsilon = q_x $ and $ 2 r_y
+ \varepsilon + 1 = q_y $. So we get
\[
\begin{array}{rcl}
    d(r,\overline{pq}) &\!\!\!\!=\!\!\!\!& \frac{q_y r_x - q_x r_y}{\sqrt{q_x^2 + q_y^2}} \\
    &\!\!\!\!=\!\!\!\!& \frac{q_y \frac{q_x+\varepsilon}{2} - q_x \frac{q_y-\varepsilon-1}{2}}{\sqrt{q_x^2 + q_y^2}}\\
    &\!\!\!\!=\!\!\!\!& \frac{1}{2\sqrt{q_x^2 + q_y^2}} (q_y \varepsilon + q_x \varepsilon + q_x ).
\end{array}
\]
Similarly,
\[
  d(r',\overline{pq})
  = \frac{1}{2\sqrt{q_x^2 + q_y^2}} ( q_y \varepsilon + q_x \varepsilon - q_x ).
\]
Setting $ x := q_x/q_y$, we observe
\[
\begin{array}{rcl}
    H(\overline{pq},\overline{prr'q}) &\!\!\!\!\leq\!\!\!\!& \frac{1}{2\sqrt{(xq_y)^2 + q_y^2}} ( q_y + 2 x q_y ) \\
    &\!\!\!\!=\!\!\!\!& \frac{x+1/2}{ \sqrt{x^2+1} } \\
    &\!\!\!\!\leq\!\!\!\!& \sqrt{5}/2,
\end{array}
\]
as the function
\[
f(x) = \frac{x+1/2}{\sqrt{x^2+1}}
\]
attains its global maximum at $x=2$.
\end{proof}

The following lemma is a statement about the order $\prec$ and will be the key ingredient of the proof of Theorem~\ref{t:main}.

\begin{lemma} \label{lemma_sym}
Let $\{x\in\Z | A \leq x < B\}$ be an interval of integers with the following properties:
\begin{itemize}
\item[(i)] Its number of elements is $B - A = 2^{k+1} - 1 $ for some number $k$.
\item[(ii)] $|(A+B-1)/2|_2 \geq k$ and $|x|_2 < k$, for any other $A \leq x < B$, $x \neq (A+B-1)/2$.
\end{itemize}
Let $x_1 \prec x_2 \prec \ldots \prec x_{2^{k+1} - 1}$ be the elements of the interval sorted
in increasing order according to $\prec$.

Then elements from the left half and elements from the right half of the interval alternate,
that is, if $x_i < (A+B-1)/2$ for some $1 \leq i < 2^{k+1} - 1$, then $x_{i+1} \geq (A+B-1)/2$.
\end{lemma}

\begin{proof}
Define $M:=(A+B-1)/2$. Assume for a contradiction that there is an $1 \leq i < 2^{k+1} - 1$ such that
both $x_i < M$ and $x_{i+1} < M$.
First we look at the case that $|A-1|_2 < |M|_2$.
Then on one hand, $x_i \prec x_i + 2^k$, because $|x_i-j|_2 = |x_i+2^k-j|_2$ for all $0 \leq j \leq x_i - A$
and $|x_i-j|_2 = |A-1|_2 < |M|_2 = |x_i+2^k-j|_2$ for $j=x_i-(A-1)$. (We use the simple observation that if $|x|_2 < k$ then $|x+2^k|_2 = |x|_2$.)
But on the other hand, $x_i + 2^k \prec x_{i+1}$, because there is a $0 \leq j_0 \leq x_i - A$ such that $|x_i+2^k-j_0|_2 = |x_i-j_0|_2 < |x_{i+1}-j_0|_2$ and $|x_i+2^k-j|_2 = |x_i-j|_2 = |x_{i+1}-j|_2$ for all $0 \leq j < j_0$. So $x_i \prec x_i + 2^k \prec x_{i+1}$, a contradiction.

In the case $|A-1|_2 > |M|_2$, we can argue similarly that $x_i \prec x_{i+1}+2^k \prec x_{i+1}$. (Note that equality never occurs, as $|A-1|_2 = |M|_2$ implies $|(A-1+M)/2|_2 \geq |M|_2 \geq k$, but $|(A-1+M)/2|_2 < k$ by assumption.)

We have shown that if $x_i < M$, then $x_{i+1} \geq M$. Similarly we can show that $x_i > M$ implies $x_{i+1} \leq M$. So the elements to the left and
to the right of $M$ alternate.
\end{proof}

\begin{proof}(of Theorem~\ref{t:main}) Let $p,q \in \Z^2$. We may assume that
$p_x < q_x$ and $p_y < q_y$. We are going to prove that $H(\overline{pq},S_\prec(p,q))
\leq 2c\log(p_x+p_y-q_x-q_y) $ for $c=\sqrt{5}/2$. Let $r \in S_\prec(p,q)$ be the
point with the property that $r_x + r_y$ is the greatest element of the segment
interval, that is, $ r_x + r_y \succ s $ for all $s \in [p_x+p_y, q_x+q_y)$, $s \neq r_x + r_y$,
see Figure~\ref{fig_splitting_at_greatest} for an example. Now let $s'$ be the
second greatest element of the segment interval according to $\prec$. Define
$k:=|s'|_2+1$.

We can extend the segment $S_\prec(p,q)$ over both endpoints, moving both $p$ and $q$
 such that $ |p_x + p_y -
1|_2 \geq k $ and $ |q_x + q_y|_2 \geq k$, that is, we extend the segment as
far as we can, so that $k$, defined as above, remains unchanged.
From Lemma \ref{lemma:subs} we get that
by this extension we decreased the Hausdorff distance by
at most a factor of $2$. Now the segment interval contains exactly $2^{k+1}-1$
elements and $r_x+r_y$ is the element in the very middle. We call such a segment \emph{normalized}.

We are going to proceed by induction on $k$ to prove that for all normalized
digital line segments $ H(\overline{pq},S_\prec(p,q)) \leq ck $ with $c = \sqrt{5}/2$.
This will prove the theorem, as $k + 1 = \log(p_x+p_y-q_x-q_y + 1)$, the distance
of the unnormalized original segment (we started from) is at most $2ck = 2c(\log(p_x+p_y-q_x-q_y +
1) - 1) \leq 2c\log(p_x+p_y-q_x-q_y)$.

\begin{figure}
\centering
\includegraphics{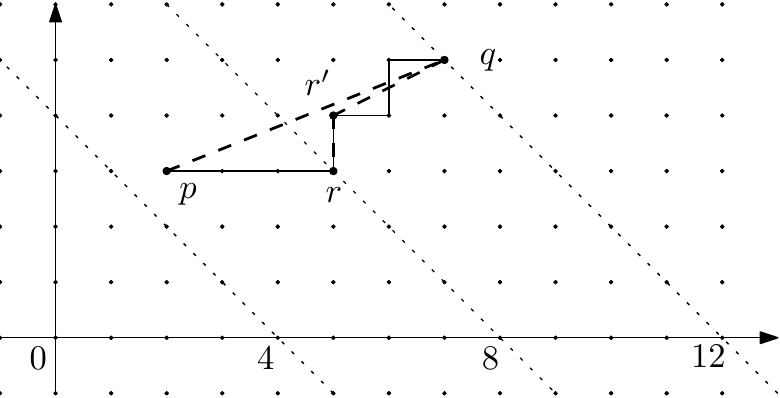}
\caption{\label{fig_splitting_at_greatest}
The digital line segment from $p = (2,3)$ to $q = (7,5)$, with the
ordered segment interval $8 \succ 10 \succ 6 \succ 9 \succ 5 \succ 11 \succ 7$.
The Euclidean segments $\overline{pq}$ and $\overline{prr'q}$ (dashed) can have
Hausdorff distance at most $c$, as shown in Lemma~\ref{lemma_c}. }
\end{figure}

In the base case $k = 1$, the segment interval consists of $3$ numbers,
so $S_\prec(p,q)$ is a path of length $3$ and by checking all possibilities we
see that $H(\overline{pq},S_\prec(p,q)) < c$.

If $k > 1$, the idea is to split the segment at $r$ into two subsegments which
are similar in some sense and apply induction. Let $r'=(r_x,r_y+1)$ be the point
that comes after $r$ in the segment $S_\prec(p,q)$. (We know that we go up at
$r$, because we go up at least once and $r_x + r_y$ is the greatest element of
the segment interval). Consider the subsegments $S_\prec(p,r)$ and $S_\prec(r',q)$ and
partition the segment interval accordingly. The key observation is that picking
the elements of the interval according to $\prec$ starting with the
greatest, we first get $r$, and then alternately an element of the left and the right
subsegment interval. This is shown in Lemma~\ref{lemma_sym} setting $A = p_x + p_y$,
$B = q_x + q_y$.
Therefore, up to a
difference of at most one, half of the $q_y-p_y-1$ greatest elements (after $r_x + r_y$) belong to $[p_x+p_y,r_x+r_y)$ and half
of them to $[r_x+r_y+1,q_x+q_y)$. This implies that $p,q,r,r'$ meet the conditions
of Lemma \ref{lemma_c}, leading to
$H(\overline{pq},\overline{prr'q}) \leq c$.
By the
induction hypothesis we have $H(\overline{pr},S_\prec(p,r)) \leq c(k-1)$ and
$H(\overline{r'q},S_\prec(r',q)) \leq c(k-1)$.
Now
$H(\overline{prr'q},S_\prec(p,q)) =
\max\{H(\overline{pr},S_\prec(p,r)),H(\overline{r'q},S_\prec(r',q))\} \leq c(k-1)$. Using
the triangle inequality we conclude $H(\overline{pq},S_\prec(p,q)) \leq ck$.
\end{proof}

\section{A step towards a characterization of CDSes} \label{s:char}

Now we approach the same problem from a different angle, taking arbitrary CDSes and trying to
find some common patterns in their structure. Knowing that condition (C1) holds for all CDSes,
it is easy to verify that we can analyze the segments with non-positive and non-negative
slopes separately, as they are completely independent. More precisely, the union of any CDS on
segments with non-positive slopes and another CDS on segments with non-negative slopes is
automatically a CDS. Having this in mind, in this section we will proceed with the analysis of
only one half of a CDS, namely of segments with non-negative slope.

%\dom{i would add here that this does not work in higher dimensions}

We will show that, in a CDS, all the segments with non-negative slope emanating
from a fixed point must be derived from a total order. However, as we will
show later, these orders may differ for different points.

%Earlier, we used that
%these orders cannot be arbitrary when we proved Case 2 of (S3) in the proof of Theorem 3.
%Remember that when we defined $S_{\prec}(p,q)$ for $p_x\le q_x$ in
%Section~\ref{s:def}, we started from $p$ and went towards $q$. Note that we
%could have done just the opposite, start from $q$ and go towards $p$. This way
%we can define the rays emanating from $p$ to both directions using only the
%total order that belongs to $p$.

%%\begin{proposition}
%%For any CDS and for any point $p\in \Z^2$, there is a total order $\prec$ on
%%$\Z$ such that the segments with non-negative slope emanating from $p$ are
%%derived from $\prec$ (in the way described in Section~\ref{s:def}).
%%\end{proposition}

%\tobi{Of course, there was nothing wrong with the proposition the way it was, but the thing is that this total order is not unique. So I think we should rather phrase it as follows, because this is what really happens. (and it confused me quite a bit when tried to understand your discussions about it)}

\begin{theorem} \label{prop_order_at_point}
For any CDS and for any point $p = (p_x,p_y) \in \Z^2$, there is a total order $\prec_p$ that
is uniquely defined on both $(-\infty,p_x+p_y-1]$ and $[p_x+p_y,+\infty)$, such that the
segments with non-negative slope emanating from $p$ are derived from $\prec_p$ (in the way
described in Section~\ref{s:def}).\end{theorem}

\begin{proof}
%%We fix a CDS $\cal S$ and a point $p$. Obviously, if there is a total order defining all the segments with non-negative slope emanating from $p=(p_x,p_y)$, then the suborder induced on integers smaller than $p_x+p_q$ is relevant for segments for which $p$ is the upper-right endpoint, and the suborder induced on the remaining integers is relevant for segments for which $p$ is the lower-left endpoint. Hence, those two sets of segments can be observed separately. Knowing this, we will look just at the latter type of segments in the remainder of the proof.
We fix a CDS $\cal S$ and a point $p$. The segments with non-negative slope with
$p=(p_x,p_y)$ as their upper-right point will induce an order on the integers
smaller than $p_x+p_y$, and the segments for which $p$ is the lower-left endpoint will
induce an order on the rest of the integers. In the following we will just look at
the latter type of segments. First, we prove an auxiliary statement.

\begin{lemma} \label{lemma_aux}
In a CDS, it cannot happen that for two segments with non-negative
slope having the same lower-left endpoint $p$, one of them goes up at $(a,C-a)$
and the other goes right at $(b,C-b)$, for some $C$ and $a>b$.
\end{lemma}

\begin{proof} (of Lemma~\ref{lemma_aux})
Let us, for a
contradiction, assume the opposite, see Figure~\ref{fig_induced_order}.
Now, we look at the $a-b+1$ segments between the point $p$
and each of the points on the line $x+y=C$ between the points $(a,C-a)$ and
$(b,C-b)$. It is possible to extend all of them through their upper-right
endpoints, applying (S4). Note that each of the extended segments goes through a
different point on the line $x+y=C$, and hence, because of condition (C3), no two
of them can go through the same point on the line $x+y=C+1$. But, there are only
$a-b$ available points on the line $x+y=C+1$ between the points $(a,C-a+1)$ and
$(b+1,C-b)$, one less than the number of segments, a contradiction.
\end{proof}

\begin{figure}
\centering
\includegraphics{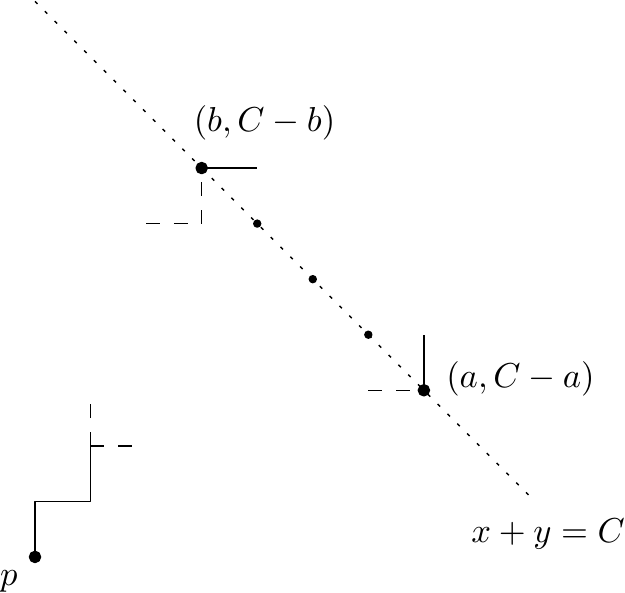}
\caption{Two segments having $p$ as their lower-left endpoint, one of them going up at $(a,C-a)$
and the other going right at $(b,C-b)$, with $a>b$.\label{fig_induced_order}}
\end{figure}

%Now consider the order $\prec_p$ defined as suggested in Section~\ref{s:def},
%Milos: I commented this out, as we don't define it this way, we define it "sort of" the opposite way, but that's too complicated to write...
Now, we define the relation $\prec_p$ in the following way.
Whenever there is a segment in $\cal S$ with non-negative slope starting at $p$, going right at a point $(x,D-x)$
and going up at a point $(x',E-x')$, for some $x$ and $x'$, we set $D \prec_p E$.
This way we defined a relation on $[p_x+p_y,+\infty)$, which is obviously irreflexive. To show that $\prec_p$ is a total
order, it remains to prove that it is asymmetric, transitive and total.

To show asymmetry, assume for a contradiction that for some integers $D$ and $E$ we
have both $D \prec_p E$ and $E \prec_p D$. That can happen
only when there are two segments with non-negative slope having $p$ as their
lower-left endpoint, such that on the line $x+y=D$ one of them goes up, the
other right, and then on $x+y=E$ they both go in different direction than at
$x+y=D$. But then the situation described in Lemma~\ref{lemma_aux} must occur
on one of the two lines, a contradiction.

Next, if $ C \prec_p D $ and $ D \prec_p E $, then we also have $ C \prec_p E $ -- we just take a
segment starting from $p$ that goes right at $C$ and up at $D$, and (if necessary) extend it
until it passes the line $x+y=E$. It must also go up at $E$, because of $ D \prec_p E $ and
the asymmetry of $\prec_p$. Hence, the relation $\prec_p$ is transitive.

It remains to prove the totality of $\prec_p$. That is, for any pair of integers $p_x + p_y \leq D < E$, either $D \prec_p E$
or $E \prec_p D$ holds. Consider a segment from $p$ to some point $q$
on the line $x+y = E$, such that this segment splits at $q$, that is, there
are two extensions of the segment, one going up and another one going right.
Such segment exists since in the upper-right quadrant of $p$, the line $x+y = E + 1$ contains one more point than
the line $x+y = E$. If we look at all the segments between $p$ and the points on $x+y = E + 1$, the pigeonhole
principle ensures that two of them will contain the same point $q$ on the line $x+y = E$.
%(There is one more segment starting at $p$ going to a point on the line $x+y = E + 1$ than there is
%going to a point on the line $x+y = E$, so exactly one of the segments must split.)
Now the segment $S(p,q)$ crosses the line $x+y=D$ at some point $q'$.
Depending on whether it goes up or right at the point $q'$, either $E \prec_p D$
or $D \prec_p E$ holds.
\end{proof}

%\tobi{I changed the following paragraph a little bit, \emph{first} introducing the CDS and \emph{then} explaining the order it induces if we look at a point $(a,b)$, because that's more like I understand the logic behind it.}

To see that these orders can differ for different points, consider the following
example of a CDS, which we call the \emph{waterline example} because of the
special role of the $x$-axis. To connect two points with a segment, we do the following.
Above the $x$-axis we go first right, then up, below the
$x$-axis we go first up, then right, and when we have to traverse the $x$-axis,
we go straight up to it, then travel on it to the right, and finally continue up, see
Figure~\ref{fig_waterline_example}. It is easy to check that this construction
satisfies all five axioms.

\begin{figure}
\centering
\includegraphics{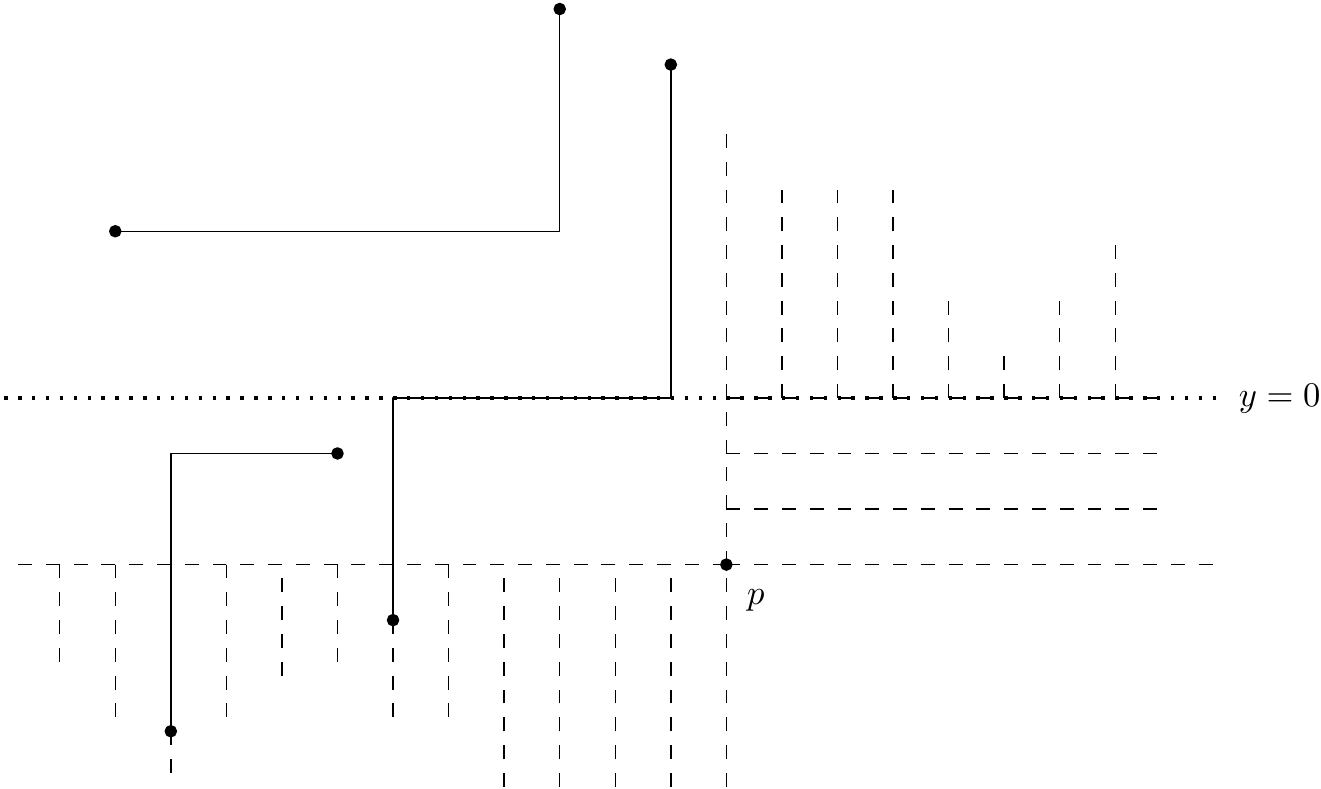}
\caption{\label{fig_waterline_example}
The waterline example: Examples of three characteristic segments, and the rays
emanating from a point $p$, which is below the waterline
%If you start above the $x$-axis go first to the right and then up. If you start and end below,
%go first up then right. If you have to cross the $x$-axis, go first up, then right on the
%$x$-axis and then up again.
}
\end{figure}

Now, if we consider a point $p=(a,b)$ below the waterline, $b < 0$, the
induced total order $\prec_p$ on $[a+b,+\infty)$
is $a \prec_p a+1 \prec_p \ldots \prec_p (+\infty) \prec_p a-1 \prec_p a-2 \prec_p \ldots \prec_p a+b $,
and the order on $(-\infty,a+b-1]$ is $ a+b-1 \prec_p a+b-2 \prec_p \ldots \prec_p -\infty$.
If $p$ is above the waterline, $b \geq 0$, the induced total orders are $a-1 \prec a-2 \prec \ldots \prec (-\infty) \prec a \prec a+1 \prec \ldots \prec a+b-1 $ and $ a+b \prec a+b+1 \prec \ldots \prec +\infty$.
Obviously, there is no total order on $\Z$ compatible with these orders for all possible choices of $p$.

%\mSt{Why don't we also put the order we get for $b\geq 0$? Even though the case $b<0$ is enough to prove our point, I would still give both orders.}
%\tobi{Ok, I included it, please check if I got it right.}
% Milos: It's ok.

The special role played by the $x$-axis in the waterline example
can be fulfilled by any other monotone digital line with a positive slope;
above the line go right, then up, below the line go up, then right,
and whenever the line is hit, follow it until either the $x$- or the $y$-coordinate
matches that of the final destination, see Figure~\ref{fig_exotic_example}.
Again, it is straightforward to show that this way we obtain a CDS.

%Note that in this case,

%\mSt{Did anyone forget to finish the sentence here?}
%\tobi{I don't know, I think it wasn't me ...}
% dom: it could have been me but I have no cloe what I wanted to add, so let's leave it like this.

\begin{figure}
\centering
\includegraphics{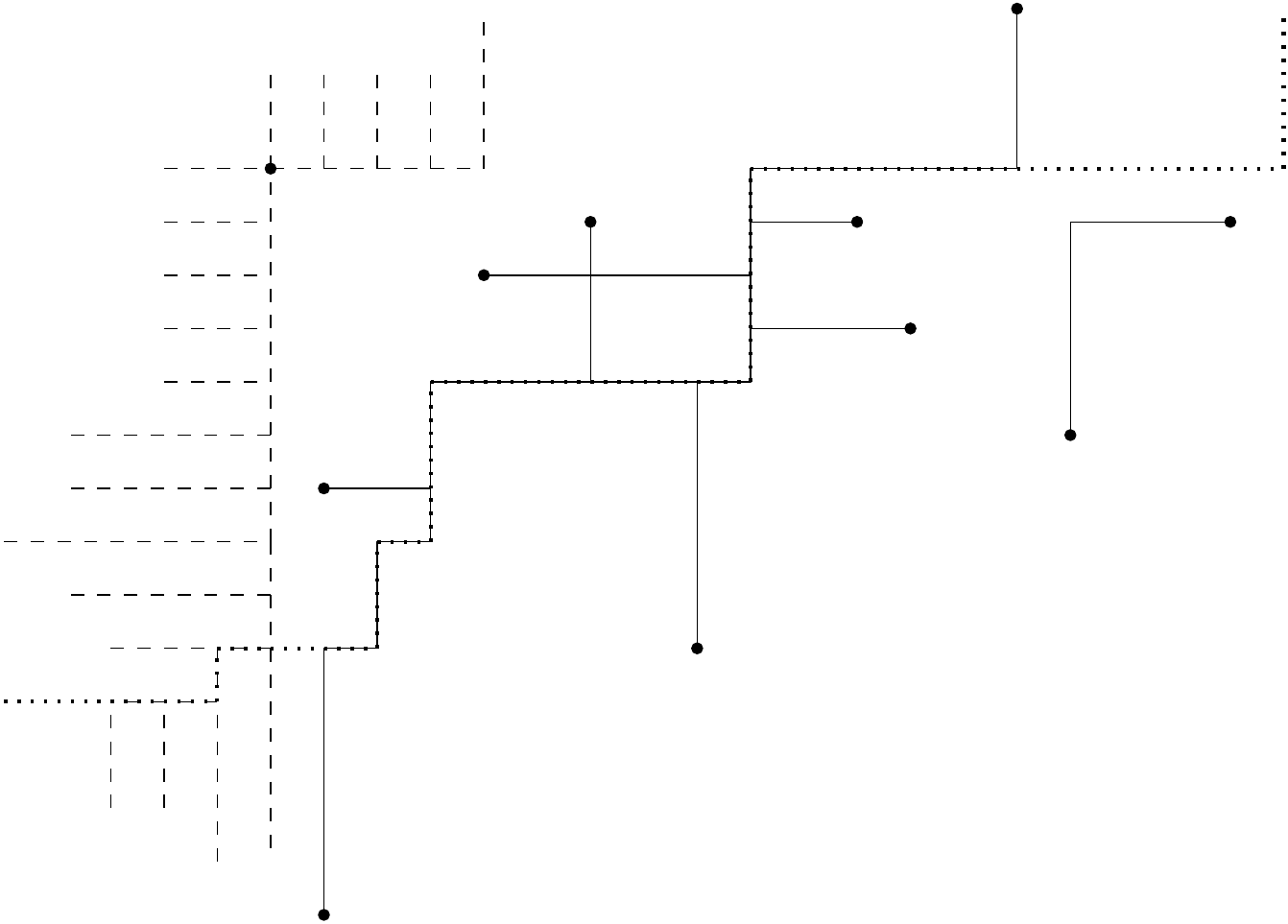}
\caption{A more exotic example of a CDS with an arbitrary ``special'' line (bold and dotted).
\label{fig_exotic_example}}
\end{figure}

%A similar construction also works if the distinguished line has a big negative slope; go first right, then up, except when you hit the line, then go up as much as you can before you leave it. If there are no two adjacent horizontal segments on the line, then this also gives a CDS\footnote{With two adjacent horizontal segments, it would violate (S4), the prolongation property.} (see figure ???).

%\tobi{Either I don't understand the last example or it isn't a CDS (the way understood, it violates the prolongation property anyway), so I commented it out.}

%\dom{is this last construction really good? I am not 100\% sure}

%\mSt{Please take care of the previous paragraph, so that all the examples are nicely written and for sure not false. :)}

A way to see that such a CDS cannot be derived from a total order is to
observe that now the diagonal translation of a digital line segment does not always
yield another digital line segment, that is, these examples do not satisfy the condition from
Observation~\ref{obs_diagonal_translation_invariance}.
Actually, we can prove that it suffices to add this condition
to force a unifying total order on all integers.

\begin{theorem}
    Let $\cal S$ be a CDS, such that for any $t \in \Z$ and any $p,q \in \Z^2$,
\[
    S(p+(t,-t),q+(t,-t)) = S(p,q) + (t,-t).
\]
    Then there is
    a unique total order $\prec$ on $\Z$ such that $\cal S = \cal S_{\prec}$.
\end{theorem}

\begin{proof}
Let $p,q \in \Z^2$. By Theorem~\ref{prop_order_at_point} the segments starting at $p$
define a unique total order $\prec_p$ on $[p_x+p_y,\infty)$ and similarly the segments
starting at $q$ define a unique total order $\prec_q$ on $[q_x+q_y,\infty)$.
Let $q_x + q_y \geq p_x + p_y$. We want to check whether these two orders agree on
$[q_x+q_y,\infty)$. Assume for contradiction there are integers $q_x+q_y \leq A < B$ such
that there is a segment $S$ starting at $p$ going to some point $r$
implying $B \prec_p A$ and another segment $T$ starting at $q$ implying $A \prec_q B$, see Figure~\ref{fig_unique_order} for an example.
We translate the segment $S$ diagonally by a vector $(t,-t)$ until $q$ lies on the translated segment $S'$
from $p'$ to $r'$. Then the subsegment from $q$ to $r'$ still goes up at level $A$ and to the right at level $B$, implying
$B \prec_q A$, a contradiction.

Therefore we can define a unique total order $\prec$ as follows. To compare two integers $A$ and $B$ take any point $p$
with $p_x + p_y \leq A$ and define $A \prec B$ if and only if $A \prec_p B$. By the argument above, this definition is
independent of the choice of $p$ and the arguments in the proof of Theorem~\ref{prop_order_at_point} carry over to
$\prec$, which shows it is a total order.
\end{proof}

\begin{figure}
\centering
\includegraphics{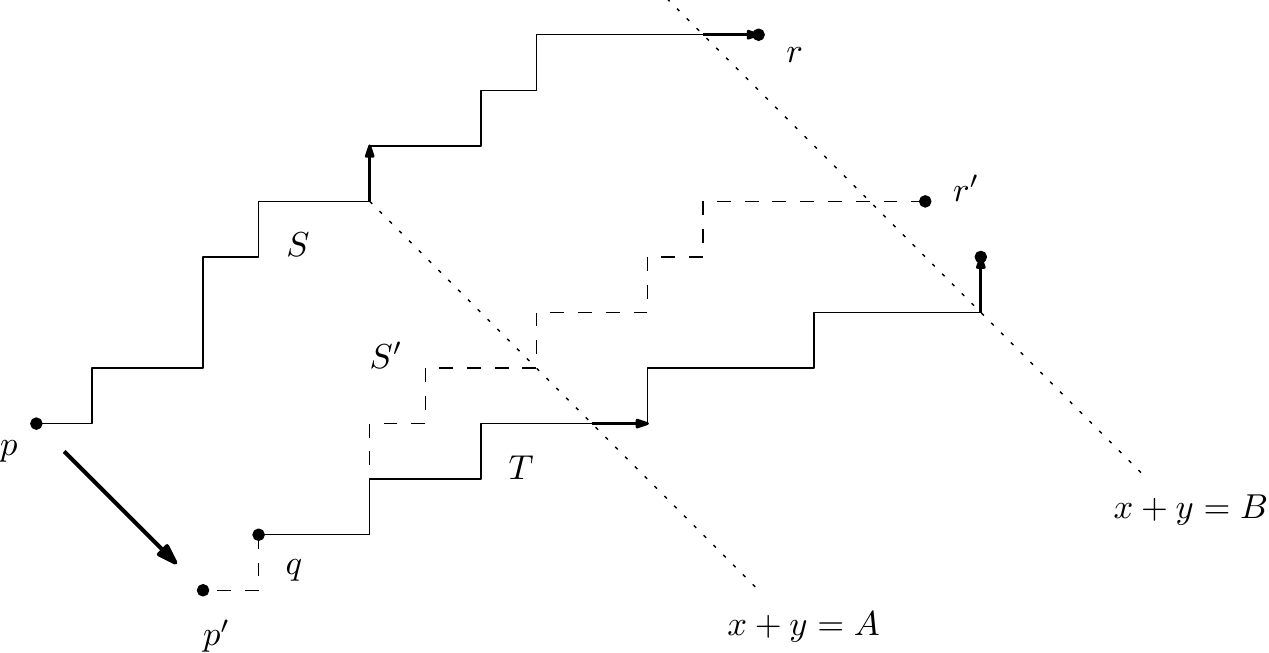}
\caption{Two line segments leading to a inconsistency in the derived orders. After translating one of them diagonally (dashed)
we find a contradiction to the subsegment property.
\label{fig_unique_order}}
\end{figure}

%\mSt{I commented out the "Higher dimensions" part. We don't really have anything to say there, at least at this point. We can leave it for the full-length paper, and then think about it more.}
% \dom{I think they are detailed enough as this is just a sidenote. I also think they are correct but would feel better is someone else also checked it. I agree with the higher dims.}

\section{Digital lines}\label{s:lines}

Even though our focus is on the digitalization of line segments, the present setup can be
conveniently extended to a definition of digital lines. We say that a \emph{digital line} is the
vertex set of a path infinite in both directions in the $\Z^2$ base graph, such that the
digital line segment between any two points on the digital line belongs to the digital line.

In this section we restrict our attention to CDSes that are derived from a total order as
described in Section \ref{s:def}. Furthermore, for simplicity, we are only going to consider
lines with non-negative slope, meaning that the Euclidean segment between any two points of
the line has non-negative slope (including zero and infinity).

Consider a digital line $\ell$ derived from ${\cal S}_\prec$. We define $A_{\ell} \subseteq
\Z$ to be the set of numbers $x+y$ for which $\ell$ goes upward at $(x,y)$ and call it the
\emph{slope} of $\ell$. Note that the slope $A_{\ell}$ is an interval in $(\Z,\prec)$ which is
unbounded in the increasing direction, i.e., if $x \in A_{\ell}$, then for any $y \succ x$ we
have $y \in A_{\ell}$. This implies that there is a natural total order on the set of possible
slopes given by inclusion.

The following observation follows directly from the definition of $S_{\prec}$. Starting at a point $p$ we
can construct a line with slope $A$.

\begin{obs}
Every line $\ell$ can be described by its slope $A_{\ell}$ and a point $p$ it contains. Also,
given a point $p$, every $\prec$-interval $A$ which is unbounded in the increasing direction
is a valid slope of a line through $p$, that is, there is a line $\ell$ such that $p \in \ell$ and
$A_{\ell} = A$.
\end{obs}

It follows directly from Theorem~\ref{theorem_order_to_CDS} that two lines having a point in
common cannot split and then meet later. Using similar arguments we can see that two lines
cannot ``touch'' without crossing.

\begin{obs} \label{obs_dont_touch}
If two different lines
intersect, then they either cross (having a common segment), or they
have a common half-line.
\end{obs}

It is straightforward to verify that the following lemma holds.

\begin{lemma} \label{lemma_slope_difference}
Consider two different slopes $A$ and $B$ with $A \subset B$. Let $I = B \setminus A$ be the
difference of the slopes. We distinguish three cases.
\begin{enumerate}
\item[1)] If $I$ is finite, then
there are lines $l$ and $s$ such that $A_l=A$, $A_s=B$ and $l$ and $s$
intersect in a lower-left halfline and there are lines $l'$, $s'$ with slopes
$A_{l'}=A$, $A_{s'}=B$ intersecting in an upper-right halfline.

\item[2)] If $I$ is infinite and bounded in one direction in the natural order on $\Z$, then we can
find lines $l$ and $s$ with slope $A$ and $B$ intersecting in a lower-left or an upper-right
halfline depending on whether $I$ is lower- or upper-bounded with respect to the natural order.
\item[3)] If $I$ is unbounded in both directions, then all lines $l$ and $s$ with slope $A$ and $B$ do intersect
in a finite segment.

\end{enumerate}
\end{lemma}

We define two lines to be \emph{parallel} if they do not cross, that is, according
to Observation~\ref{obs_dont_touch}, two lines are parallel if they are disjoint or if they agree on a halfline.
We distinguish two possible cases how the slope $A_{\ell}$ of a line $\ell$ may look like. If there is a $c \in \Z$
such that $A_{\ell} = [c,\infty)_{\prec} = \{a\in\Z | a \succ c \} \cup \{c\}$ or $A_{\ell} =
(c,\infty)_{\prec} = \{a\in\Z | a \succ c\}$, that is, if its boundary can be described by the smallest
element, either including or excluding this element, then we call the slope $A_{\ell}$
\emph{rational}. In the special cases $A_{\ell} = \Z$ and $A_{\ell} = \emptyset$ we define
$A_{\ell}$ as rational, too. If no such $c$ exists, or, in other words, if $A_{\ell}$ does not
have a smallest element and its complement $\Z \setminus A_{\ell}$ does not have a greatest
element (and they are not empty), then we call $A_{\ell}$ \emph{irrational}.

We proceed by analyzing the digital lines derived from the special total order that we defined in Section~\ref{s:H}.
With the help of the previous lemma, we can characterize which lines do have unique parallels and which do not.

%\dom{why do we need this special order? maybe we could put here two statements, one about arbitrary orders and another about the power of 2 order}

%\tobi{I explicitly mention what we use of the order in the proof now. Actually, the statement about arbitrary orders is the lemma above. Of course, we could generalize the following theorem by just asking for a dense total order with the additional property that every non-trivial interval is unbounded in both directions in the natural oder. but does this make sense? I prefer it this way. Now by reading the proof one immeadiately can tell what the general statement would be, but we don't have to introduce another property of orders ...}

\begin{theorem}
Let $\prec$ be the total order on $\Z$ defined in Section~\ref{s:H} using the powers of $2$,
$\cal S_{\prec}$ the CDS obtained from it, and $\ell$ a line with respect to $\cal S_{\prec}$. Let $p
\in \Z^2$ such that $p \notin \ell$.
\begin{enumerate}
\item[1)] If the slope $A_{\ell}$ is irrational, then there is
a unique line $\ell'$ through $p$ that is parallel to $\ell$. Furthermore,
$A_{\ell'} = A_{\ell}$.
\item[2)]
If $A_{\ell}$ is rational, then $\ell$ has exactly two parallels
$\ell'$ and $\ell''$ through $p$, one of them with the same slope as $\ell$, the other
with a slope that differs by one element. Consequently $\ell'$ and $\ell''$ either have
a common halfline to the left and split at one point to run parallel at distance one thereafter, or
the other way round.
\end{enumerate}
\end{theorem}

\begin{proof}
Note that $\prec$ is a dense order on $\Z$, that is, for any integers $a \prec b$, we find a
$c \in \Z$ such that $a \prec c \prec b$. Equivalently, this means that every $\prec$-interval
$[a,b]_{\prec}$ with $a \neq b $ is infinite. Let $s$ be an arbitrary line through $p$
parallel to $\ell$. Consider the symmetric difference of the slopes $I = A_{\ell} \triangle
A_{s} $. If $I$ is empty, the slopes are equal and $s$ is the diagonal translate of $\ell$
going through $p$. If $I \neq \emptyset$, by density of the order, $I$ is either infinite or
consists only of one element $a$. If $I$ is infinite, then it is unbounded according to the
natural order in both directions. (Beside density, this is the only property of $\prec$ that
we use.) According to Lemma~\ref{lemma_slope_difference}, $s$ and $\ell$ do intersect in a
finite segment in this case, which is a contradiction because $s$ is parallel to $\ell$.
Therefore, $I$ consists of one element. If $A_{\ell}$ is irrational, by adding one element to
$A_{\ell}$ or removing one element from $A_{\ell}$ we do not get an interval. Hence, in this
case, there is only one possible slope for any line parallel to $\ell$, namely, $A_{\ell}$. If
$A_{\ell}$ is rational, then either $A_{\ell} = \{a\in\Z | a \succ c \} \cup \{c\}$ or
$A_{\ell} = \{a\in\Z | a \succ c\}$, for some $c\in \Z$, so removing $c$ or adding $c$,
respectively, yields the only different possible slope for a parallel through $p$. As the
slopes of these two parallels only differ by one element, they either have a common halfline
to the left or a common halfline to the right.
\end{proof}

\section{Higher dimensions} \label{s:high-d}

It is natural to ask whether there is a CDS in more than two dimensions. The definition of a
CDS directly carries over to the higher dimensional spaces. Instead of $\Z^2$ we now consider
$\Z^d$, for a fixed integer $d \geq 3$, with the usual graph structure, that is, two points $p$
and $q$ are adjacent if they differ in exactly one coordinate by exactly one. The axioms
(S1)-(S4) stay the same, verbatim, and the monotonicity axiom (S5) now reads as follows: If
for $p = (p_1,\ldots,p_d), q=(q_1,\ldots,q_d) \in \Z^d$ there is an $i$ such that $p_i = q_i$,
then for all $r=(r_1,\ldots,r_d) \in S(p,q)$ we have $r_i=p_i=q_i$. We define the \emph{slope
type} of a pair of points $(p,q)$ as the sign vector $(\sigma_1,\ldots,\sigma_d) \in
\{+1,-1\}^d$, $\sigma_i = +1$ if $p_i \leq q_i$, and $\sigma_i = -1$ if $p_i > q_i$. The slope
type of a digital segment $S(p,q)$ is defined as the slope type of $(p,q)$. We say a segment
has \emph{strictly positive slope}, if its slope type is $(+1,\ldots,+1)$, that is, the
coordinates are monotone increasing in each coordinate.

It is not hard to see that we can again derive a consistent system from an arbitrary total
order $\prec$ on $\Z$ if we only consider segments that have strictly positive slope, that is,
slope type $(+1,\ldots,+1)$; the only difference is that now we have to cut the segment
interval $[p_1+\ldots+p_d,q_1+\ldots+q_d-1]$ into $d$ parts, according to $\prec$. Let $p,q
\in \Z^d$ be two points, such that $p_i \leq q_i$ for all $1\leq i \leq d$. We can define the
segment $S_{\prec}(p,q)$ in a similar way as in the two dimensional case -- starting at $p$
and repeatedly going in one of the $d$ possible directions, collecting points from $\Z^d$,
until reaching $q$. If we are at a point $(r_1,\ldots,r_d)$ for which $r_1+r_2+\ldots+r_d$ is
among the $q_d-p_d$ greatest elements of the segment interval
$[p_1+\ldots+p_d,q_1+\ldots+q_d-1]$ according to $\prec$, we proceed in direction $d$; if it
is among the $q_{d-1} - p_{d-1}$ greatest elements of the segment interval that remain after
removing the $q_d-p_d$ elements that were the greatest, we proceed in direction $d-1$, and so
on. Finally, if $r_1+r_2+\ldots+r_d$ it is among the $q_1 - p_1$ smallest elements of the
segment interval, we proceed in the first direction.

\begin{theorem} \label{theorem_order_to_CDS_3d}
The definition above yields a CDS of segments with strictly positive slope.
\end{theorem}

%\dom{if you don't mind, i don't like the proof for the reason that it only works for 3
%dimensions. since there is nothing special about three, it should be emphasized in the proof
%that it works for all dimensions. i would use 1,2,3 instead of u,l,r and then it is enough to
%say that wlog $a_1>b_1$ and thus we must have an $a_i<b_i$ which is a contradiction.}
%
%\tobi{I agree, I changed everything to arbitrary dimension and I think it is still
%understandable, the proof even got easier because I got rid of a stupid case distinction that
%was necessary the way I put it before ...}

\begin{proof}
The crucial axiom to verify is the subsegment property (S3). Assume there are two segments
with strictly positive slope, splitting at some point $p$ and meeting again for the first time
at $q$. The two subsegments from $p$ to $q$ can be seen as words over the alphabet
$[d]=\{1,2,\ldots,d\}$, where $1$ stands for going in the first direction, $2$ for going in
the second direction, and so on. So let $\alpha=a_1 a_2 \ldots a_k, \beta = b_1 b_2 \ldots b_k
\in [d]^k$ be these two words. By assumption they differ at the beginning and at the end and
they contain the same number of each of the letters, that is, for any $l \in [d]$, $|\{i:~a_i
= l\}| = |\{i:~b_i = l\}|$. Without loss of generality we may assume that $a_1 > b_1$. Then
there must be an $1 < i \leq k$ such that $a_i < a_1$ and $b_i > b_1$. (If there were no such
$i$, then for any $1 < i \leq k$ with $a_i < a_1$ we have $b_i \leq b_1 < a_1$. Now looking at
all letters that are strictly smaller than $a_1$ in both words, we see that there is at least
one such letter more in $b$, namely at the first position, a contradiction.) This leads to a
contradiction, as translated back to the original setting it implies $p_1 + \ldots + p_d \succ
p_1 + \ldots + p_d + i -1$, if we look at the interval of the first segment, and at the same
time $p_1 + \ldots + p_d \prec p_1 + \ldots + p_d + i -1$, if we look at the interval of the
second segment. This proves (S3).

The rest of the axioms can be verified similarly as in the 2-dimensional case.
\end{proof}

Of course, we can use the same construction to define segments for all the remaining slope
types. However, unlike in the $2$-dimensional case, putting them all together in an attempt to
construct a complete CDS fails, as the axiom (S3) is violated. We are curious if this approach
to construction can be modified to yield a CDS.

\section*{Acknowledgments}
We are indebted to Ji\v{r}\'{\i} Matou\v{s}ek, who posed this problem at the 7th Gremo Workshop on Open Problems - GWOP 2009. We would like to
thank J\' ozsef Solymosi for participating in the fruitful discussions that led to the main result of this paper. Finally, we thank the
organizers of GWOP 2009 for inviting us to the workshop and providing us with a gratifying working environment.

\bibliographystyle{abbrv}
%%\bibliography{segments}

\end{document}